\documentclass{extarticle}
\usepackage{geometry}
\title{\textbf{Bribery Can Get Harder in Structured Multiwinner Approval
Election}}
\author{%
  Bartosz Kusek,\textsuperscript{\rm 1}
  Robert Bredereck,\textsuperscript{\rm 2}
	Piotr Faliszewski,\textsuperscript{\rm 1}\\
  Andrzej Kaczmarczyk,\textsuperscript{\rm 1}
  Du{\v{s}}an Knop\textsuperscript{\rm 3}
  \vspace{0.3cm}\\
	\begin{minipage}{.9\textwidth}
		\centering\normalsize
    \textsuperscript{\rm 1} AGH University, Krak\'{o}w, Poland\\
    \textsuperscript{\rm 2} Institut für Informatik, TU Clausthal,
	  Clausthal-Zellerfeld, Germany\\
    \textsuperscript{\rm 3} Czech Technical University in Prague, Prague, Czech
	  Republic
  \end{minipage}\vspace{0.3cm}\\
	\begin{minipage}{.9\textwidth}
		\centering\small
		\texttt{kussy.kusy@gmail.com, robert.bredereck@tu-clausthal.de,}\\
		\texttt{faliszew@agh.edu.pl, andrzej.kaczmarczyk@agh.edu.pl,}\\
	  \texttt{dusan.knop@fit.cvut.cz}
  \end{minipage}
}
\date{}

\usepackage[noend,ruled,vlined,linesnumbered]{algorithm2e}

\usepackage{url}
\usepackage{natbib} %

\usepackage{graphicx}
\usepackage{amsmath}
\usepackage{amsfonts}
\usepackage{amsthm}
\usepackage{booktabs}
\usepackage{multirow}
\usepackage{thmtools} 
\usepackage{thm-restate}
\usepackage{wrapfig}

\usepackage{subcaption}
\usepackage{bm}

\usepackage{dsfont}
\usepackage{mathtools}

\usepackage{comment}
\usepackage{pgfplots}
\pgfplotsset{compat=newest}
\usepackage{cleveref}
\usepackage{subcaption}
\usetikzlibrary{matrix}
\usepackage[textsize=scriptsize]{todonotes}

\newcommand{\myparagraph}[1]{\smallskip \noindent\textbf{#1}\;}

\usepackage{nicefrac}

\newtheorem{theorem}{Theorem}

\newtheorem{definition}{Definition}
\newtheorem{example}{Example}

\newtheorem{remark}{Remark}

\newcommand{\removememaybe}[1]{}

\newcommand{\calS}{\mathcal{S}}
\newcommand{\calT}{\mathcal{T}}
\newcommand{\np}{{\mathrm{NP}}}

\newcommand{\p}{{\mathrm{P}}}

\DeclareMathOperator{\score}{score}
\DeclareMathOperator{\cost}{cost}
\DeclareMathOperator{\avscore}{score}

\newcommand{\pprec}{\mathrm{prec}}
\newcommand{\ssucc}{\mathrm{succ}}

\newcommand{\Winner}{\textsc{Winner}}
\newcommand{\bordaOWAwinner}[1][]{%
	$\beta$%
	\ifx&#1&%
	\else
	(#1)%
	\fi%
	-\Winner%
}%

\allowdisplaybreaks

\providecommand{\np}{\ensuremath{\mathrm{NP}}}

\providecommand{\p}{\ensuremath{\mathrm{P}}}

\providecommand{\score}{{\mathrm{score}}}
\providecommand{\cost}{{\mathrm{cost}}}
\providecommand{\av}{{\mathrm{AV}}}

\newcommand{\scoreof}[1]{\ensuremath{s(#1)}}

\newcommand{\first}{}    
\def\first/{red}
\newcommand{\second}{}    
\def\second/{blue}
\newcommand{\third}{}    
\def\third/{black}
\newcommand{\fourth}{}    
\def\fourth/{green}
\clearpage{}%

\usepackage{etoolbox}
\newcommand{\appsymb}{$\lozenge$}

\newcommand{\appref}[1]{{\appsymb}}

\begin{document}
\maketitle
\begin{abstract}
We study the complexity of constructive bribery in the context of structured
multiwinner approval elections.  Given such an election, we ask
whether a certain candidate can join the winning committee by
adding, deleting, or swapping approvals, where each such action
comes at a cost and we are limited by a budget. We assume our
elections to either have the candidate interval or the voter
interval property, and we require the property to hold also after
the bribery. While structured elections usually make manipulative
attacks significantly easier, our work also shows examples of
the opposite behavior. We conclude by presenting preliminary insights regarding
the destructive variant of our problem.\bigskip

\end{abstract}

  \section{Introduction}\label{sec:intro}

  We study the complexity of bribery under the multiwinner approval
  rule, in the case where the voters' preferences are structured.
	Specifically, we use the bribery model of
	\citet{fal-sko-tal:c:bribery-measure}, where one can either add, delete, or
	swap approvals, and we consider \emph{candidate interval} and \emph{voter
	interval} preferences~\citep{elk-lac:c:approval-sp-sc}.

  In multiwinner elections, %
  the voters express their preferences over the available candidates
  and
  use this information to select a winning committee (i.e., a
  fixed-size subset of candidates).
  We focus on one of the simplest and most common scenarios, where each
  voter specifies the approved candidates, %
  and those with
  the highest numbers of approvals form the committee. Such elections
  are used, e.g., to choose city councils, boards of trustees, or to
  shortlist job candidates. Naturally, there are many other rules and
  scenarios, but they do not appear in practice as often as this
  simplest one.  For more details on multiwinner voting, we point the
  readers to the overviews of
  \citet{fal-sko-sli-tal:b:multiwinner-voting} and
  \citet{lac-sko:t:approval-survey}.

  In our scenario, we are given an election, including the contents of
  all the votes, and, depending on the variant, we can either add,
  delete, or swap approvals, but each such action comes at a cost. Our
  goal is to find a cheapest set of actions that ensure that a given
  candidate joins the winning committee.
  Such problems, where we modify the votes to ensure a certain
  outcome, are known under the umbrella name of \emph{bribery}, and
  were first studied by \citet{fal-hem-hem:j:bribery}, whereas our specific variant
  is due to \citet{fal-sko-tal:c:bribery-measure}. Historically, bribery
  problems indeed aimed to model vote buying, but currently more
	benign interpretations prevail. For example,
	\citet{fal-sko-tal:c:bribery-measure} suggest using
  the cost of bribery as a measure of a candidate's success:
	A~candidate who did not win, but can be put into the committee at a low cost,
	certainly did better than one whose bribery is expensive. In particular, since
	our problem is used for post-election analysis, it is natural to assume that
	we know all the votes. For other similar interpretations, we point, e.g., to
	the works of \citet{xia:c:margin-of-victory},
	\citet{shi-yu-elk:c:robustness}, \citet{bre-fal-kac-nie-sko-tal:j:robustness},
	\citet{boe-bre-fal-nie:c:counting-swap-bribery}, or
	\citet{bau-hog:c:counting-bribery}.
  \citet{fal-rot:b:control-bribery} give a more general overview of bribery problems.

  We assume that our elections either satisfy the candidate interval
  (CI) or the voter interval (VI)
  property~\citep{elk-lac:c:approval-sp-sc}, which %
  correspond to the classic notions of
  single-peakedness~\citep{bla:b:polsci:committees-elections} and
  single-crossingness~\citep{mir:j:single-crossing,rob:j:tax} from the
  world of ordinal elections. Briefly put, the CI property means that
  the candidates are ordered and each voter approves some interval of
  them, whereas the VI property requires that the voters are ordered
  and each candidate is approved by some interval of voters. For
  example, the CI assumption can be used to model political elections,
  where the candidates appear on the left-to-right spectrum of
  opinions and the voters approve those, whose opinions are close
  enough to their own.  Importantly, we require our elections to have
  the CI/VI property also after the bribery; this approach is standard
  in bribery problems with structured
  elections~\citep{bra-bri-hem-hem:j:sp2,men-lar:c:bribery-sp-hard,elk-fal-gup-roy:c:swap-shift-sp-sc},
  as well as in other problems related to manipulating election
  results~\citep{wal:c:uncertainty-in-preference-elicitation-aggregation,fal-hem-hem-rot:j:single-peaked-preferences,fit-hem:c:ties-sp-manipulation-bribery}
  (these references are examples only).

  \begin{example}\label{ex:1}
    Let us consider a hotel in a holiday resort. The hotel has its
    base staff, but each month it also hires some additional help. For
    the coming month, the expectation is to hire extra staff for $k$
    days. Naturally, they would be hired for the days when the hotel
    is most busy (the decision to request additional help is made a
    day ahead, based on the observed load; the $k$ days do not need to be consecutive).  Since hotel bookings are
    typically made in advance, one knows which days are expected to be
    most busy.  However, some people will extend their stays, some
    will leave early, and some will have to shift their stays.  Thus
    the hotel managers would like to know which days are likely to
    become the busiest ones after such changes: Then they could inform
    the extra staff as to when they are expected to be needed, and
    what changes in this preliminary schedule might happen.
    Our bribery
    problem (for the CI setting) captures exactly the problem that the
    managers want to solve: The days are the candidates, $k$ is the
    committee size, and the bookings are the approval votes (note that each booking must regard a consecutive set of days). Prices of
    adding, deleting, and moving approvals correspond to the
    likelihood that a particular change actually happens (the managers usually know which changes are more or less likely). 
    Since the bookings must be consecutive,  the election has
    to have the CI property also after the bribery. The managers can
    solve such bribery problem for each of the days and see
    which ones can most easily be among the $k$ busiest ones. 

    Note that the value of $k$ can be estimated from previous experience, be limited by the budget for the salaries for the additional workers, depend on the
    predicted workload, be subject to the availability of the additional workers, and the like. The managers may even want to solve several instances of the problem, with different values of $k$.
  \end{example}

  \begin{example}\label{ex:2}
    For the VI setting, let us consider a related scenario.
    There is a team of archaeologists who booked a set of excavation
    sites, each for some consecutive number of days (they work on several sites in parallel). The team may want to add
    some extra staff to those sites that require most working 
    days. However, as in the previous example, the bookings might get
    extended or shortened. The team's manager may use bribery to evaluate
    how likely it is that each of the sites becomes one of the most
    work-demanding ones. In this case, the days are the voters, and the
    sites are the candidates.
  \end{example}

  There are two main reasons why structured elections are studied.
  Foremost, as in the above examples, sometimes they simply capture the
  exact problem at hand.
  Second, many problems that are intractable in general, become
  polynomial-time solvable if the elections are structured. Indeed,
  this is the case for many $\np$-hard winner-determination
  problems~\citep{bet-sli-uhl:j:mon-cc,elk-lac:c:approval-sp-sc,pet-lac:j:spoc}
  and for various problems where the goal is to make some candidate a
  winner~\citep{fal-hem-hem-rot:j:single-peaked-preferences,mag-fal:j:sc-control},
  including some bribery
  problems~\citep{bra-bri-hem-hem:j:sp2,elk-fal-gup-roy:c:swap-shift-sp-sc}. There
  are also some problems that stay intractable even for structured
  elections~\citep{elk-fal-sch:j:fallback-shift,yan:c:borda-control-sp}\footnote{These
    references are not complete and are meant as examples.}
  as well as examples of complexity reversals,
  where assuming structured preferences turns a polynomial-time
  solvable problem into an intractable one. However, such reversals are 
  rare and, to the best of our knowledge, so far were only observed by
  \citet{men-lar:c:bribery-sp-hard}, for the case
  of weighted elections with three candidates (but see also the work
  of \citet{fit-hem:c:ties-sp-manipulation-bribery}, who find
  complexity reversals that stem from replacing total ordinal votes
  with ones that include ties).  \smallskip

  \myparagraph{Our Contribution.}  We provide an almost complete
  picture of the complexity of bribery by either adding, deleting, or
  swapping approvals under the multiwinner approval voting rule, for
  the case of CI and VI elections, assuming either that each bribery
  action has identical unit price or that they can be priced
  individually (see Table~\ref{tab:results}). By comparing our results
  to those for the unrestricted setting, provided by
  \citet{fal-sko-tal:c:bribery-measure}, we find that any
  combination of tractability and intractability in the structured and
  unrestricted setting is possible. For example:
  \begin{enumerate}
  \item Bribery by adding approvals is solvable in polynomial time
    irrespective if the elections are unrestricted or have the CI or VI
    properties.
  \item Bribery by deleting approvals (where each deleting action is
    individually priced) is solvable in polynomial time in the
    unrestricted setting, but becomes $\np$-hard for CI elections (for
    VI ones it is still in $\p$).
  
  \item Bribery by swapping approvals only to the designated candidate (with
    individually priced actions) is $\np$-hard in the unrestricted
    setting, but becomes polynomial-time solvable both for CI and VI
    elections.
  
  \item Bribery by swapping approvals (where each action is individually
    priced and we are not required to swap approvals to the designated
    candidate only) is $\np$-hard in each of the considered settings.
  \end{enumerate}
  We largely focus on the constructive setting, where the goal is to
  ensure that some candidate belongs to at least one winning committee
  (indeed, all the results above are for this setting). However, we
  also give a glimpse of what happens in the destructive setting,
  where we want to ensure that a given candidate does not belong to
  any winning committees. In a certain sense, in this case we also
  observe a form of ``reversal.'' Typically, destructive variants of
  bribery (and related) problems are at least as easy to work with
  as their constructive counterparts, and lead to more positive
  results. In our case---albeit this is mostly an intuitive
  feeling---the situation for CI elections is the opposite. Obtaining 
  the results for the destructive setting seems more challenging and
  leads to less satisfying theorem statements (e.g., we need less
  appealing prices) than in the constructive setting.\smallskip

  \noindent\textbf{Possibility of Complexity Reversals.}\;
  So far, most of the problems studied for structured elections were
  subproblems of the unrestricted ones. For example, a winner
  determination algorithm that works for all elections, clearly also
  works for the structured ones and complexity reversal is impossible.
  The case of bribery is different because, by assuming structured
  elections, not only do we restrict the set of possible inputs, but
  also we constrain the possible actions.  Yet, scenarios where
  bribery is tractable are rare, and only a handful of papers
	considered bribery in structured domains (we mention those of
	\citet{bra-bri-hem-hem:j:sp2}, \citet{fit-hem:c:ties-sp-manipulation-bribery},
	\citet{men-lar:c:bribery-sp-hard},
	\citet{elk-fal-gup-roy:c:swap-shift-sp-sc}), so opportunities for
  observing complexity reversals were, so far, very limited. We show
  several such reversals, obtained for very natural settings.

  \begin{table}
    \centering
    \setlength{\tabcolsep}{3pt}
    \scalebox{0.8}{
    \begin{tabular}{r|cc|cc|cc}
      \toprule
                             & \multicolumn{2}{c|}{Unrestricted} & \multicolumn{2}{c|}{Candidate Interval (CI)}        & \multicolumn{2}{c}{Voter Interval (VI)}                                                                                                               \\
      {\small (prices)}
                             & \multicolumn{1}{c}{\small (unit)}
                             & \multicolumn{1}{c|}{\small (any)}
                             & \multicolumn{1}{c}{\small (unit)}
                             & \multicolumn{1}{c|}{\small (any)}
                             & \multicolumn{1}{c}{\small (unit)}
                             & \multicolumn{1}{c}{\small (any)}                                                                                                                                                                  \\
      \midrule
      \textsc{AddApprovals}       & $\p$                                           & $\p$                                           & $\p$                                           & $\p$                                           & $\p$                                           & $\p$                                           \\[-3pt]
                             & \multicolumn{2}{c|}{\scriptsize{Faliszewski et al.~\citeyearpar{fal-sko-tal:c:bribery-measure}}} 
                             & \scriptsize{[Thm.~\ref{thm:add-approvals-ci}]} & \scriptsize{[Thm.~\ref{thm:add-approvals-ci}]} & \scriptsize{[Thm.~\ref{thm:add-approvals-vi}]} & \scriptsize{[Thm.~\ref{thm:add-approvals-vi}]} \\
      \midrule
      \textsc{DelApprovals}    & $\p$ & $\p$ & \multirow{2}{*}{?}                                              & $\np$-com.                                     & $\p$                                           & $\p$                                           \\
      & \multicolumn{2}{c|}{\scriptsize{Faliszewski et al.~\citeyearpar{fal-sko-tal:c:bribery-measure}}}& 
      & \scriptsize{[Thm.~\ref{thm:delete-approvals-ci}]} & \scriptsize{[Thm.~\ref{thm:delete-approvals-vi}]} & \scriptsize{[Thm.~\ref{thm:delete-approvals-vi}]} \\
      \midrule
      \textsc{SwapApprovals} & $\p$ & $\np$-com. & $\p$                                           & $\p$                                           & $\p$                                           & $\p$                                           \\[-3pt]
      \textsc{to p}                       &\multicolumn{2}{c|}{\scriptsize{Faliszewski et al.~\citeyearpar{fal-sko-tal:c:bribery-measure}}}& \scriptsize{[Thm.~\ref{thm:swap-approvals-to-p-ci}]} & \scriptsize{[Thm.~\ref{thm:swap-approvals-to-p-ci}]} & \scriptsize{[Thm.~\ref{thm:swap-approvals-to-p-vi}]} & \scriptsize{[Thm.~\ref{thm:swap-approvals-to-p-vi}]} \\[2mm]
      \textsc{SwapApprovals}      & $\p$ & $\np$-com. & $\np$-com.                                      & $\np$-com.                                      & \multirow{2}{*}{?}                                              & $\np$-com.                                     \\[-3pt]
                             &\multicolumn{2}{c|}{\scriptsize{Faliszewski et al.~\citeyearpar{fal-sko-tal:c:bribery-measure}}}& \scriptsize{[Thm.~\ref{thm:swap-approvals-ci}]} & \scriptsize{[Thm.~\ref{thm:swap-approvals-ci}]} &  & \scriptsize{[Thm.~\ref{thm:vi-swaps}]} \\
      \bottomrule
    \end{tabular}}
  
    \caption{Our results for the CI and VI domains,
      together with those of \protect\citet{fal-sko-tal:c:bribery-measure} for
      the unrestricted setting. \textsc{SwapApprovals to p} refers to the
      problem where each action has to move an approval to the preferred
      candidate.}
    \label{tab:results}
  \end{table}

\section{Preliminaries}

For a positive integer $t$, we write $[t]$ to mean the set
$\{1, \ldots, t\}$. By writing $[t]_0$ we mean the set
$[t] \cup \{0\}$.

\myparagraph{Approval Elections.}
An \emph{approval election} $E = (C,V)$ consists of a set of
candidates $C = \{c_1, \ldots, c_m\}$ and a collection of voters
$V = \{v_1, \ldots, v_n\}$.  Each voter $v_i \in V$ has \emph{an
  approval ballot} (or, equivalently, \emph{an approval set}) which
contains the candidates that $v_i$ approves.  We write $v_i$ to refer
both to the voter and to his or her approval ballot; the meaning
will always be clear from the context.

A \emph{multiwinner voting rule} is a function $f$ that given an
election $E = (C,V)$ and a committee size $k \in [|C|]$ outputs a
nonempty family of winning committees (where each committee is a
size-$k$ subset of $C$).
We disregard the issue of tie-breaking and assume all winning
committees to be equally worthy,
i.e., we adopt the nonunique winner model.

Given an election $E = (C,V)$, we let the approval score of a candidate
$c \in C$ be the number of voters that approve~$c$, and we denote it as
$\avscore_E(c)$. The approval score of a committee $S \subseteq C$ is
$\avscore_E(S) = \sum_{c \in S} \avscore_E(c)$.  Given an election~$E$
and a committee size~$k$, the \emph{multiwinner approval voting} rule,
denoted $\av$, outputs all size-$k$ committees with the highest
approval score.
Occasionally we also consider the \emph{single-winner approval rule},
which is defined in the same way as its multiwinner variant, except
that the committee size is fixed to be one. For simplicity, in this
case we assume that the rule returns a set of tied winners (rather
than a set of tied size-$1$ winning committees).

\myparagraph{Structured Elections.}  We focus on elections where the
approval ballots satisfy either the \emph{candidate interval (CI)} or
the \emph{voter interval (VI)}
properties~\citep{elk-lac:c:approval-sp-sc}:
\begin{enumerate}
\item An election has the CI
  property (is a CI election) if there is an ordering of the candidates
  (called the \emph{societal axis}) such that each approval ballot
  forms an interval with respect to this ordering.

\item An election has the VI %
  property (is a VI election) if there is an ordering of the voters such
  that each candidate is approved by an interval of the voters (for
  this ordering).
\end{enumerate}
Given a CI election, we say that the voters have CI ballots or,
equivalently, CI preferences; we use analogous conventions for the VI
case. As observed by \citet{elk-lac:c:approval-sp-sc}, there are
polynomial-time algorithms that test if a given election is CI or VI
and, if so, provide appropriate orders of the candidates or voters;
these algorithms are based on solving the \emph{consecutive ones}
problem~\citep{boo-lue:j:consecutive-ones-property}.

\myparagraph{Notation for CI Elections.}
Let us consider a candidate set $C = \{c_1, \ldots, c_m\}$ and a
societal axis $\rhd = c_1 c_2 \cdots c_m$.  Given two candidates
$c_i, c_j$, where $i \leq j$, we write $[c_i,c_j]$ to denote the
approval set $\{c_i, c_{i+1}, \ldots, c_j\}$.

\myparagraph{Bribery Problems.}
We focus on the variants of bribery in multiwinner approval elections
defined by~\citet{fal-sko-tal:c:bribery-measure}.
Let $f$ be a multiwinner voting rule and let \textsc{Op} be one of
\textsc{AddApprovals}, \textsc{DelApprovals}, and
\textsc{SwapApprovals} operations (in our case $f$ will either be
$\av$ or its single-winner variant). In the $f$-\textsc{Op-Bribery}
problem we are given an election $E = (C,V$), a committee size~$k$, a
preferred candidate $p$, and a nonnegative integer $B$ (the
budget). We ask if it is possible to perform at most $B$ unit
operations of type \textsc{Op}, so that $p$ belongs to at least one
winning committee (this is the constructive variant of the problem; requiring
only minor adaptions in proofs, presented in~Appendix~\ref{app:other_models},
all our results for this variant also hold if we require that $p$ belongs to all
winning committees):
\begin{enumerate}
\item For \textsc{AddApprovals}, a unit operation adds a given
  candidate to a given voter's ballot.
\item For \textsc{DelApprovals}, a unit operation removes a given
  candidate from a given voter's ballot.
\item For \textsc{SwapApprovals}, a unit operation replaces a given
  candidate with another one in a given voter's ballot.
\end{enumerate}
Like \citet{fal-sko-tal:c:bribery-measure}, we also study the
variants of \textsc{AddApprovals} and \textsc{SwapApprovals} problems
where each unit operation must involve the preferred candidate.

We are also interested in the priced variants of the above problems,
where each unit operation comes at a cost that may depend both on the
voter and the particular affected candidates;
we ask if we can achieve our goal by performing operations of total
cost at most~$B$. We distinguish the priced variants %
by putting a dollar sign in front of the operation type. For example,
\textsc{\$AddApprovals} means a variant where adding each candidate to
each approval ballot has an individual cost.

\myparagraph{Bribery in Structured Elections.}
We focus on the bribery problems where the elections have either the CI
or the VI property. For example, in the
\textsc{AV-\$AddApprovals-CI-Bribery} problem the input election has
the CI property (under a given societal axis) and we ask if it is
possible to add approvals with up to a given cost so that (a) the
resulting election has the CI property for the same societal axis, and
(b) the preferred candidate belongs to at least one winning committee.
The VI variants are defined analogously (in particular, the voters'
order witnessing the VI property is given %
and the election must still have the VI property with respect to this
order after the bribery).
The convention that the election must have the same structural
property before and after the bribery, and %
the fact that the
order witnessing this property is part of the input, is standard in
the literature; see, e.g., the works of
\citet{fal-hem-hem-rot:j:single-peaked-preferences},
\citet{bra-bri-hem-hem:j:sp2},
\citet{men-lar:c:bribery-sp-hard}, and
\citet{elk-fal-gup-roy:c:swap-shift-sp-sc}.
Further, it also follows naturally from some applications (as in the scenarios
from Examples~\ref{ex:1} and~\ref{ex:2}).

\myparagraph{Computational Problems.}  For a graph $G$, by $V(G)$ we
mean its set of vertices and by $E(G)$ we mean its set of edges. A
graph is cubic if each of its vertices is connected to exactly three
other ones.
Our $\np$-hardness proofs rely on reductions from variants of the
\textsc{Independent Set} and \textsc{X3C} problems, both known to be
$\np$-complete~\citep{gar-joh:b:int,gon:j:x3c}.

\begin{definition}
  In the \textsc{Cubic Independent Set} problem 
  we are given a cubic graph $G$ and an integer $h$; we ask if $G$
  has an independent set of size $h$ (i.e., a set of $h$ vertices such that no two of them are connected).
\end{definition}
\begin{definition}
  In the \textsc{Restricted Exact Cover by 3-Sets} problem (\textsc{RX3C}) we are given a universe
  $X$ 
  of $3n$ elements and a family
  $\calS$ 
  of $3n$ size-$3$ subsets of
  $X$. Each element from $X$ appears in exactly three sets from
  $\calS$. We ask if it is possible to choose $n$ sets from $\calS$
  whose union is $X$.
\end{definition}

\section{Adding Approvals}
\label{sec:adding}

For the case of adding approvals, all our bribery problems (priced and
unpriced, both for the CI and VI domains) remain solvable in
polynomial time. Yet, compared to the unrestricted setting, our
algorithms require more care. For example, in the unrestricted case it
suffices to simply add approvals for the preferred candidate~\citep{fal-sko-tal:c:bribery-measure} (choosing
the voters where they are added in the order of increasing costs for
the priced variant); a similar
approach works for the VI case, but with a different ordering of the
voters.

\begin{theorem}%
\label{thm:add-approvals-vi}
  \textsc{AV-\$AddApprovals-VI-Bribery} $\in \p$.
\end{theorem}
\begin{proof}
  Consider an input with election $E = (C,V)$, committee size $k$,
  preferred candidate $p$, and budget $B$. Without loss of generality,
  we assume that $V = \{v_1, \ldots, v_n\}$ and the order witnessing
  the VI property is $v_1 \rhd v_2 \rhd \cdots \rhd v_n$. We note that
  it is neither beneficial nor necessary to ever add approvals for
  candidates other than $p$. Let $v_i, v_{i+1}, \ldots, v_j$ be the
  interval of voters that approves $p$, and let $s$ be the lowest
  number of additional approvals that $p$ needs to obtain to become a
  member of some winning committee (note that $s$ is easily computable
  in polynomial time). Our algorithm proceeds as follows: We consider
  all nonnegative numbers $s_\ell$ and $s_r$ such that
  (a)~$s = s_\ell + s_r$, (b)~$i - s_\ell \geq 1$, and
  (c)~$j+s_r \leq n$, and for each of them we compute the cost of
  adding an approval for $p$ to voters $v_{i-s_\ell}, \ldots, v_{i-1}$
  and $v_{j+1}, \ldots, v_{j+s_r}$.  We choose the pair that generates
  lowest cost and we accept if this cost is at most $B$. Otherwise we
  reject.

  The polynomial running time follows directly. Correctness is
  guaranteed by the fact that we need to maintain the VI property and
  that it suffices to add approvals for $p$ only.
\end{proof}

The CI case introduces a different complication. Now, adding an approval for the preferred
candidate in a given vote also requires adding approvals for all 
those between him or her and the original approval set.
Thus, in addition to bounding the bribery's cost, we also need to 
track the candidates whose scores increase. %

\begin{theorem}\label{thm:add-approvals-ci}
  \textsc{AV-\$AddApprovals-CI-Bribery} $\in \p$.
\end{theorem}
\begin{proof}
  Our input consists of an election $E = (C,V)$, committee size $k$,
  preferred candidate $p \in C$, budget $B$, and the information about
  the costs of all the possible operations (i.e., for each voter and
  each candidate that he or she does not approve, we have the price
  for adding this candidate to the voter's ballot). Without loss of
  generality, we assume that
  $C = \{\ell_{m'}, \ldots, \ell_1, p, r_1, \ldots, r_{m''}\}$,
  $V = \{v_1, \ldots, v_n\}$, each voter approves at least one
  candidate,\footnote{Without this assumption we could still make our
    algorithm work.  We would guess the number of voters
    who do not approve any candidates to approve $p$ alone (we would choose these 
    voters to minimize the cost of adding these approvals). Then we would continue as in
    the proof, but knowing that none of the voters in the group can be bribed further.} and the election is CI with respect to the order:
  \[
    \rhd = \ell_{m'} \cdots  \ell_2\  \ell_1 \  p \  r_1\ 
     r_2  \cdots  r_{m''}.
  \]
  We start with a few observations. First, %
  if a voter
  already approves~$p$ then there is no point in adding any approvals
  to his or her ballot. Second, if some voter does not approve $p$,
  then we should either not add any approvals to his or her ballot, or
  add exactly those approvals that are necessary to ensure that $p$
  gets one. For example, if some voter has approval ballot
  $\{r_3, r_4, r_5\}$ then we may either choose to leave it intact or
  to extend it to $\{p, r_1, r_2, r_3, r_4, r_5\}$.
  We let $L = \{\ell_{m'}, \ldots, \ell_1\}$ and
  $R = \{r_1, \ldots, r_{m''}\}$, and we partition the voters into
  three groups, $V_\ell$, $V_p$, and $V_r$, as follows:
  \begin{enumerate}
    \item $V_p$ contains all the voters who approve $p$,
    \item $V_\ell$ contains the voters who approve members of $L$
          only,
    \item $V_r$ contains the voters who approve members of $R$ only.
  \end{enumerate}
  Our algorithm proceeds as follows (by guessing we mean iteratively
  trying all possibilities; Steps~\ref{alg1:dp1} and~\ref{alg1:dp2}
  will be described later): %
  \begin{enumerate}
    \item Guess the numbers $x_\ell$ and $x_r$ of voters from $V_\ell$
          and $V_r$ whose approval ballots will be extended to approve $p$.
    \item Guess the numbers $t_\ell$ and $t_r$ of candidates from
          $L$ and $R$ that will end up with higher approval scores
          than $p$ (we must have $t_\ell + t_r < k$ for $p$ to join a
          winning committee).
    \item \label{alg1:dp1} Compute the lowest cost of extending exactly $x_\ell$ votes
          from $V_\ell$ to approve $p$, such that at most $t_\ell$ candidates from
          $L$ end up with more than $\score_E(p)+x_\ell+x_r$ points (i.e.,
          with score higher than $p$); denote this cost as $B_\ell$.
    \item \label{alg1:dp2} Repeat the above step for the $x_r$ voters
          from $V_r$, with at most $t_r$ candidates obtaining more than
          $\score_E(p)+x_\ell+x_r$ points; denote the cost of this operation
          as $B_r$.
    \item If $B_\ell + B_r \leq B$ then accept (reject if no choice of
          $x_\ell$, $x_r$, $t_\ell$, and $t_r$ leads to acceptance).
  \end{enumerate}
  One can verify that this algorithm is correct (assuming we know how
  to perform Steps~\ref{alg1:dp1} and~\ref{alg1:dp2}).

  Next we describe how to perform Step~\ref{alg1:dp1} in polynomial
  time (Step~\ref{alg1:dp2} is handled analogously). We will need some
  additional notation. For each $i \in [m']$, let $V_{\ell}(i)$
  consist exactly of those voters from $V_\ell$ whose approval ballots
  include candidate $\ell_i$ but do not include $\ell_{i-1}$ (in other
  words, voters in $V_\ell(i)$ have approval ballots of the form
  $[\ell_j, \ell_i]$, where $j \geq i$). Further, for each
  $i \in [m']$ and each $e \in [|V_\ell(i)|]_0$ let $\cost(i,e)$ be
  the lowest cost of extending $e$ votes from $V_\ell(i)$ to approve
  $p$ (and, as a consequence, to also approve candidates
  $\ell_{i-1}, \ldots, \ell_1$).  If $V_\ell(i)$ contains fewer than
  $e$ voters then $\cost(i,e) = +\infty$. For each $e \in [x_\ell]_0$,
  we define $S(e) = \score_E(p) + e + x_r$. Finally, for each
  $i \in [m']$, $e \in [x_\ell]_0$, and $t \in [t_\ell]_0$ we define
  function $f(i,e,t)$ so that:
  \begin{enumerate}
    \item[] $f(i,e,t)$ = the lowest cost of extending exactly $e$ votes
          from $V_\ell(1) \cup \cdots \cup V_\ell(i)$ (to approve~$p$) such
          that at most $t$ candidates among $\ell_1, \ldots, \ell_i$ end up
          with more than $S(e)$ points (function~$f$ takes value $+\infty$
          if doing so is impossible).         
  \end{enumerate}
  Our goal in Step~\ref{alg1:dp1} of the main algorithm is to compute
  $f(m',x_\ell,t_\ell)$, which we do via dynamic programming.  To this
  end, we observe that the following recursive equation holds (let
  $\chi(i,e)$ be $1$ if $\score_E(\ell_i) > S(e)$ and let $\chi(i,e)$
  be $0$ otherwise; we explain the idea of the equation below):
  \begin{align*}
    f(i,e,t)\! =\! \min_{e' \in [e]_0} \! \big( \cost(i,e') + f(i-1, e-e', t - \chi(i,e))\big).
  \end{align*}
  The intuition behind this equation is as follows. We consider each
  possible number $e' \in [e]_0$ of votes from $V_\ell(i)$ that can
  be extended to approve $p$. The lowest cost of extending the votes
  of $e'$ voters from $V_\ell(i)$ is, by definition,
  $\cost(i,e')$. Next, we still need to extend $e-e'$ votes from
  $V_\ell(i-1), \ldots, V_\ell(1)$ and, while doing so, we need to
  ensure that at most $t$ candidates end up with at most $S(e)$
  points. Candidate $\ell_i$ cannot get any additional approvals from
  voters $V_\ell(i-1), \ldots, V_\ell(1)$, so he or she exceeds this
  value exactly if $\score_E(\ell_i) > S(e)$ or, equivalently, if
  $\chi(i,e) = 1$. This means that we have to ensure that at most
  $t - \chi(i,e)$ candidates among $\ell_{i-1}, \ldots, \ell_1$ end up
  with at most $S(e)$ points. However, since we extend $e'$ votes from
  $V_\ell(i)$, we know that candidates $\ell_{i-1}, \ldots, \ell_1$
  certainly obtain $e'$ additional points (as compared to the input
  election). Thus we need to ensure that at most $t - \chi(i,e)$ of
  them end up with score at most $S(e-e')$ after extending the votes
  from $V_\ell(1) \cup \ldots \cup V_\ell(i-1)$. This is ensured by
  the $f(i-1,e-e',t-\chi(i,e))$ component in the equation (which also
  provides the lowest cost of the respective operations).

  Using the above formula, the fact that $f(1,e,t)$ can be computed
  easily for all values of~$e$ and~$t$, and standard dynamic
  programming techniques, we can compute $f(m',x_\ell,t_\ell)$ in
  polynomial time. This suffices for completing Step~\ref{alg1:dp1} of
  the main algorithm and we handle Step~\ref{alg1:dp2}
  analogously. Since all the steps of can be performed in polynomial
  time, the proof is complete.
\end{proof}

Both above theorems also apply to the cases where we can only add
approvals for the preferred candidate. The algorithm from
Theorem~\ref{thm:add-approvals-vi} is designed to do just that, and
for the algorithm from Theorem~\ref{thm:add-approvals-ci} we can set
the price of adding other approvals to be $+\infty$.

\section{Deleting Approvals}
\label{sec:deleting}

The case of deleting approvals is more intriguing. Roughly speaking, in the unrestricted 
setting it suffices  to delete approvals from sufficiently many
candidates  that
have higher scores than~$p$, for whom doing so is least expensive~\citep{fal-sko-tal:c:bribery-measure}. 
The same general strategy works for the VI case because we still can delete approvals for different candidates independently.

\begin{theorem}%
\label{thm:delete-approvals-vi}
  \textsc{AV-\$DelApprovals-VI-Bribery} $\in~\p$.
\end{theorem}

\begin{proof}
  Let our input consist of an election $E=(C,V)$, preferred candidate
  $p \in C$, committee size~$k$, and budget~$B$. We assume that
  $V = \{v_1, \ldots, v_n\}$ and the election is VI with respect to
  ordering the voters by their indices.  Let $s = \score_E(p)$ be the
  score of $p$ prior to any bribery. We refer to the candidates with
  score greater than~$s$ as \emph{superior}.
  Since it is impossible to increase the score of $p$ by deleting
  approvals, we need to ensure that the number of superior candidates
  drops to at most $k-1$.

  For each superior candidate $c$, we compute the lowest cost for
  reducing his or her score to exactly~$s$. Specifically, for each
  such candidate $c$ we act as follows.
  Let $t = \score_E(c) - s$ be the number of $c$'s approvals that we need to delete and
  let $v_{a}, v_{a+1}, \ldots, v_{b}$ be the interval of voters that approve $c$.  
  For each $i \in [t]_0$ we compute the cost of deleting $c$'s approvals  among
  the first $i$ and the last $t-i$ voters in the interval (these are the only operations that achieve our goal and maintain the VI property of the election); we store the lowest of these costs as ``the cost of $c$.''
  
  Let $S$ be the number of superior candidates (prior to any bribery).
  We choose $S-(k-1)$ of them with the lowest costs. If the sum of these
  costs is at most $B$ then we accept and, otherwise, we reject.
\end{proof}

For the CI case, our problem turns out to be $\np$-complete.
Intuitively, the reason for this is that in the CI domain deleting an
approval for a given candidate requires either deleting all the
approvals to the left or all the approvals to the right on the
societal axis.
Indeed, our main trick is to introduce approvals that must be deleted
(at zero cost), but doing so requires choosing whether to delete their
left or their right neighbors (at nonzero cost).  This result is our
first example of a complexity reversal.

\begin{theorem}\label{thm:delete-approvals-ci}
  \textsc{AV-\$DelApprovals-CI-Bribery} is $\np$-complete.
\end{theorem}

\begin{proof}
  We give a reduction from \textsc{RX3C}. Let $I=(X,\mathcal{S})$ be
  the input instance, where $X = \{x_1,\ldots, x_{3n}\}$ is the
  universe and $\mathcal{S} = \{S_1,\ldots,S_{3n}\}$ is a family of
  size-$3$ subsets of $X$. By definition, each element of $X$ belongs
  to exactly three sets from $\mathcal{S}$.
  We form an %
  instance of 
  \textsc{AV-\$DelApprovals-CI-Bribery}
  as follows.

  We have the preferred candidate $p$, for each universe element~$x_i \in X$
  we have corresponding universe candidate $x_i$, for each
  set $S_j \in \calS$ we have set candidate $s_j$, and we have set $D$
  of $2n$ dummy candidates (where each individual one is denoted
  by~$\diamond$).  Let $C$ be the set of just-described $8n+1$
  candidates and let $S = \{s_1, \ldots, s_{3n}\}$ contain the set
  candidates. We fix the societal axis to be:
  \[
      \rhd = 
      \underbrace{\overbrace{s_1 \cdots s_{3n}}^{3n}
      \overbrace{\diamond \cdots \diamond}^{2n}
      \overbrace{x_1 \cdots x_{3n}}^{3n}
      p}_{8n+1}
  \]
  Next, we form the voter collection $V$:
  \begin{enumerate}
  \item For each candidate in $S \cup D \cup \{p\}$, we have two
    voters that approve exactly this candidate. We refer to them as
    the \emph{fixed voters} and we set the price for deleting their
    approvals to be~$+\infty$. We refer to their approvals as
    \emph{fixed}.

  \item For each set $S_j = \{x_a,x_b,x_c\}$, we form three
    \emph{solution voters}, $v(s_j,x_a)$, $v(s_j,x_b)$, and
    $v(s_j,x_c)$, with approval sets $[s_j,x_a]$, $[s_j, x_b]$, and
    $[s_j,x_c]$, respectively. For a solution voter $v(s_i,x_d)$, we
    refer to the approvals that $s_i$ and $x_d$ receive as 
    \emph{exterior}, and to all the other ones as 
    \emph{interior}. The cost for deleting each exterior approval is
    one, whereas the cost for deleting the interior approvals is
    zero. Altogether, there are $9n$ solution voters.
  \end{enumerate}
  To finish the construction, we set the committee size $k=n+1$ and
  the budget $B=9n$.  Below, we list the approval scores prior to any
  bribery (later we will see that in successful briberies one always
  deletes all the interior approvals):
  \begin{enumerate}
  \item $p$ has $2$ fixed approvals,
  \item each universe candidate has $3$~exterior approvals (plus some
    number of interior ones),
  \item each set candidate has $3$ exterior approvals and $2$ fixed
    ones (plus some number of interior ones), and
  \item each dummy candidate has $2$ fixed approvals (and $9n$
    interior ones).
  \end{enumerate}

  We claim that there is a bribery of cost at most $B$ that ensures
  that $p$ belongs to some winning committee if and only if $I$ is a
  yes-instance of \textsc{RX3C}. For the first direction, let us
  assume that $I$ is a yes-instance and let $\calT$ be a size-$n$
  subset of $\calS$ such that $\bigcup_{S_i \in \calT} S_i = X$ (i.e.,
  $\calT$ is the desired exact cover). We perform the following
  bribery: First, for each solution voter we delete all his or her
  interior approvals. Next, to maintain the CI property (and to lower
  the scores of some candidates), for each solution voter we delete
  one exterior approval.  Specifically, for each set
  $S_j = \{x_a,x_b,x_c\}$, if $S_j$ belongs to the cover (i.e., if
  $S_i \in \calT$) then we delete the approvals for $x_a$, $x_b$, and
  $x_c$ in %
  $v(s_j,x_a)$, $v(s_j,x_b)$, and $v(s_j,x_c)$, respectively;
  otherwise, i.e., if $S_j \notin \calT$, we delete the approvals for
  $s_j$ in these votes.  As a consequence, all the universe candidates
  end up with two exterior approvals each, the $n$ set candidates
  corresponding to the cover end up with three approvals each (two
  fixed ones and one exterior), the $2n$ remaining set candidates and
  all the dummy candidates end up with two fixed approvals
  each. Since~$p$ has two approvals, the committee size is $n+1$,
  and only $n$ candidates have score higher than $p$, $p$ belongs to
  some winning committee (and the cost of the bribery is~$B$).

  For the other direction, let us assume that there is a bribery with
  cost at most $B$ that ensures that $p$ belongs to some winning
  committee. It must be the case that this bribery deletes exactly one
  exterior approval from each solution voter. Otherwise, since there
  are $9n$ solution voters and the budget is also $9n$, some solution
  voter would keep both his or her exterior approvals, as well as all
  the interior ones.  This means that after the bribery there would be
  at least $2n$ dummy candidates with at least three points each. Then,
  $p$ would not belong to any winning committee. Thus, each solution
  voter deletes exactly one exterior approval, and we may assume that
  he or she also deletes all the interior ones (this comes at zero
  cost and does not decrease the score of $p$).

  By the above discussion, we know that all the dummy candidates end
  up with two fixed approvals, i.e., with the same score as~$p$. Thus,
  for $p$ to belong to some winning committee, at least $5n$
  candidates among the set and universe ones also must end up with at
  most two approvals (at most $n$ candidates can have score higher
  than $p$). Let $x$ be the number of set candidates whose approval
  score drops to at most two, and let $y$ be the number of such
  universe candidates. We have that: %
  \begin{align}\label{eq:1}
    0 \leq x \leq 3n&,& 0 \leq y \leq 3n&,& \text{and}&& x+y \geq 5n.
  \end{align}
  Prior to the bribery, each set candidate has five non-interior approvals
  (including three exterior approvals) so bringing his or her score to
  at most two costs three units of budget. Doing the same for a
  universe candidate costs only one unit of budget, as universe
  candidates originally have only three non-interior approvals.
  Since our total budget is $9n$, we have:
  \begin{equation}\label{eq:2}
    3x+y \leq 9n.
  \end{equation}
  Together, inequalities~\eqref{eq:1} and~\eqref{eq:2} imply that
  $x = 2n$ and $y = 3n$. That is, for each universe candidate $x_i$
  there is a solution voter $v(s_j,x_d)$ who is bribed to delete the
  approval for $x_d$ (and, as a consequence of our previous
  discussion, who is not bribed to delete the approval for $s_j$). We
  call such solution voters \emph{active} and we
  define a family:
  \[
    \calT = \{ S_j \mid s_j \text{ is approved by some active solution
      voter} \}.
  \]
  We claim that $\calT$ is an exact cover for the \textsc{RX3C}
  instance $I$. Indeed, by definition of active solution voters we
  have that $\bigcup_{S_i \in \calT} S_i = X$. Further, it must be the
  case that $|\calT| = n$. This follows from the observation that if
  some solution voter is active then his or her corresponding set
  candidate $s_j$ has at least three approvals after the bribery (each
  set candidate receives exterior approvals from exactly three
  solution voters and these approvals must be deleted if the candidate
  is to end up with score two; this is possible only if all the three
  solution voters are not active).  Since exactly $2n$ set candidates
  must have their scores reduced to two, it must be that
  $3n - |\calT| = 2n$, so $|\calT| = n$. This completes the proof.
\end{proof}

The above proof strongly relies on using $0$/$1$/$+\infty$ prices. The case of unit prices remains open and we believe that resolving it
might be quite  challenging.

\section{Swapping Approvals}
\label{sec:swapping}
In some sense, bribery by swapping
approvals is our most interesting scenario because there are cases
where a given problem has the same complexity both in the unrestricted
setting and for some structured domain (and this happens both for
tractability and $\np$-completeness), as well as cases where the
unrestricted variant is tractable but the structured one is not or the
other way round.

\subsection{Approval Swaps to the Preferred Candidate}
Let us first consider a variant of \textsc{AV-SwapApprovals-Bribery}
where each unit operation moves an approval from some candidate to the
preferred one. We call operations of this form \textsc{SwapApprovals
  to p}. In the unrestricted setting, this problem is in $\p$ for unit
prices but is $\np$-complete if the prices are arbitrary. For the CI
and VI domains, the problem can be solved in polynomial time for both
types of prices. While for the CI domain this is not so
surprising---indeed, in this case possible unit operations are very
limited---the VI case requires quite some care.

\begin{theorem}%
\label{thm:swap-approvals-to-p-ci}
  \textsc{AV-\$SwapApprovals to p-CI-Bribery} $\in \p$.
\end{theorem}
\begin{proof}
  Consider an input with CI election $E = (C,V)$, committee size $k$, preferred candidate $p$, and budget $B$. 
  W.l.o.g., we assume that $C = \{l_{m'}, \ldots, l_1, p, r_1, \ldots, r_{m''}\}$, $V = \{v_1, \ldots, v_n\}$, and the societal axis is: 
  \[
    \rhd = \ell_{m'}\ \cdots\  \ell_2\  \ell_1 \  p \  r_1\ r_2\  \cdots\  r_{m''}.
  \]
  Since unit operations must move approvals to $p$, for each voter $v_i$ exactly one of the following holds:
  \begin{enumerate}
  \item There is $t \in [m']$ such that $v_i$ has approval set  $[\ell_t,\ell_1]$ and the only possible operation is to move an approval from~$\ell_t$ to $p$ at some given cost.
  \item There is $t \in [m'']$ such that $v_i$ has approval set  $[r_1,r_t]$ and the only possible operation is to move an approval from~$r_t$ to $p$ at some given cost. 
  \item It is illegal to move any approvals for this voter.
  \end{enumerate}
  For each candidate $c \in C \setminus \{p\}$
  and each integer $x$, we let $f(c,x)$ be the lowest cost of moving
  $x$ approvals from $c$ to $p$ (we assume that $f(c,x) = +\infty$ if
  doing so is impossible). By the above discussion, we
  can compute the values of $f$  in polynomial time.
  
  Our algorithm proceeds as follows. First, we guess score
  $y \in [n]_0$ that we expect $p$ to end up with. Second, we let~$S$
  be the set of candidates that in the input election have score
  higher than $y$. For each candidate $c \in S$ we define his or her
  cost to be $f(c, \score_E(c)-y)$, i.e., the lowest cost of moving
  approvals from $c$ to $p$ so that $c$ ends up with score~$y$. Then
  we let~$S'$ be a set of~$|S|-(k-1)$ members of~$S$ with the lowest
  costs (if~$|S| \leq k-1$ then~$S'$ is an empty set).  For
  each~$c \in S'$, we perform the approval moves implied by
  $f(c,\score_E(c)-y)$. Finally, we ensure that~$p$ has~$y$ approvals
  by performing sufficiently many of the cheapest still-not-performed
  unit operations (we reject for this value of $y$ if not enough
  operations remained). If the total cost of all performed unit
  operations is at most $B$, we accept (indeed, we have just found a
  bribery that ensures that there are at most $k-1$ candidates with
  score higher than $p$ and whose cost is not too high). Otherwise, we
  reject for this value of~$y$.  If there is no $y$ for which we
  accept, we reject.
\end{proof}

Our algorithm for the VI case is based on
dynamic programming (expressed as searching for a shortest path
in a certain graph) and relies on the fact that due to the VI property
we avoid performing the same unit operations twice.

\begin{theorem}%
  \label{thm:swap-approvals-to-p-vi}
  \textsc{AV-\$SwapApprovals to p-VI-Bribery}  $\in \p$.
\end{theorem}

\begin{proof}
  Consider an instance of our problem with an election $E=(C,V)$,
  committee size~$k$, preferred candidate $p$, and budget $B$.
  Without loss of generality, we assume that $V = \{v_1, \ldots, v_n\}$
  and that the election is VI with respect to the order
  $v_1 \rhd v_2 \rhd \cdots \rhd v_n$. We also assume that $p$ has at
  least one approval (if it were not the case, we could try all
  possible single-approval swaps to $p$ to try all ways in which this
  assumption can be satisfied; if $p$ is not already a member of some winning committee then we know that $p$ needs to receive at least one approval, so this procedure is correct).
  On a high level, our algorithm proceeds as follows: We try all
  pairs of integers $\alpha$ and $\beta$ such that
  $1 \leq \alpha \leq \beta \leq n$ and, for each of them, we check if
  there is a bribery with cost at most $B$ that ensures that the
  preferred candidate is (a)~approved exactly by voters
  $v_\alpha, \ldots, v_\beta$, and (b)~belongs to some winning
  committee. If such a bribery exists then we accept and, otherwise, we
  reject.  Below we describe the algorithm that finds the cheapest
  successful bribery for a given pair $\alpha, \beta$ (if one exists).

  Let $\alpha$ and $\beta$ be fixed. Further, let $x, y$ be two
  integers such that in the original election $p$ is approved exactly
  by voters $v_x, v_{x+1}, \ldots, v_y$. Naturally, we require that
  $\alpha \leq x \leq y \leq \beta$;
  if this condition is not met then we
  drop this $\alpha$ and $\beta$.
  We let $s = \beta - \alpha + 1$ be the score that $p$ is to have
  after the bribery. We say that a candidate $c \in C \setminus \{p\}$
  is \emph{dangerous} if his or her score in the original election is
  above $s$. Otherwise, we say that this candidate is \emph{safe}. Let
  $D$~be the number of dangerous candidates. For $p$ to become a
  member of some winning committee, we need to ensure that after the
  bribery at most $k-1$ dangerous candidates still have more than $s$
  points (each safe candidate certainly has at most $s$ points).

  To do so, we analyze a certain digraph.
  For each pair of integers $a$ and $b$ such that
  $\alpha \leq a \leq b \leq \beta$ and each integer $d \in [|C|-1]_0$
  we form a node $(a,b,d)$, corresponding to the fact that there is a
  bribery after which $p$ is approved exactly by voters
  $v_a, v_{a+1}, \ldots, v_b$ and exactly $d$ dangerous candidates
  have scores above $s$. Given two nodes $u' = (a',b', d')$ and
  $u'' = (a'',b'',d'')$, such that $a' \geq a''$, $b' \leq b''$, and
  $d'' \in \{d',d'-1\}$, there is a directed edge from $u'$ to $u''$
  with weight $\cost(u',u'')$ exactly if there is a candidate $c$ such
  that after bribing voters $v_{a''}, v_{a''+1}, \ldots, v_{a'-1}$ and
  $v_{b'+1}, \ldots, v_{b''-1}, v_{b''}$ to move an approval from $c$
  to $p$ it holds that:
  \begin{enumerate}
  \item voters approving $c$ still form an interval,
  \item if $c$ is a dangerous candidate and his or her score drops to
    at most $s$, then $d'' = d'-1$, and, otherwise, $d'' = d'$, and
  \item the cost of this bribery is exactly $\cost(u',u'')$.
  \end{enumerate}
  One can verify that for each node $u = (a,b,d)$ the weight of the
  shortest path from $(x,y,D)$ to $u$ is exactly the price of the
  lowest-cost bribery that ensures that $p$ is approved by voters
  $v_a, \ldots, v_b$ and exactly $d$ dangerous candidates have
  scores above $s$ (the VI property ensures that no approval is ever
  moved twice). Thus it suffices to find a node $(\alpha, \beta, K)$
  such that $K < k$, for which the shortest path from $(x,y,D)$ is at
  most $B$. Doing so is possible in polynomial time using, e.g., the
  classic algorithm of Dijkstra.
\end{proof}

\subsection{Arbitrary Swaps}

Next, we consider the full variant of bribery by swapping approvals.
For the unrestricted domain, the problem is $\np$-complete for general
prices, but admits a polynomial-time algorithm for unit
ones~\citep{fal-sko-tal:c:bribery-measure}. For the CI domain,
$\np$-completeness holds even for the latter. %

\begin{remark}
  The model of unit prices, applied directly to the case of
  \textsc{SwapApprovals-CI-Bribery}, is somewhat unintuitive. For
  example, consider societal axis
  $c_1 \rhd c_2 \rhd \cdots \rhd c_{10}$ and an approval set
  $[c_3,c_5]$. The costs of swap operations that transform it into,
  respectively, $[c_4,c_6]$, $[c_5,c_7]$, and $[c_6,c_8]$ are $1$,
  $2$, and $3$, as one would naturally expect. Yet, the cost of
  transforming it into, e.g., $[c_8,c_{10}]$ would also be~$3$ (move
  an approval from $c_3$ to $c_8$, from $c_4$ to $c_9$, and from $c_5$
  to $c_{10}$), which is not intuitive. Instead, it would be natural
  to define this cost to be $5$ (move the interval by $5$ positions to
  the right). Our proof of Theorem~\ref{thm:swap-approvals-ci} works
  without change for both these interpretations of unit prices.
\end{remark}

\begin{theorem}%
\label{thm:swap-approvals-ci}
  \textsc{AV-SwapApprovals-CI-Bribery} is $\np$-complete.
\end{theorem}
\begin{proof}
  We give a reduction from \textsc{Cubic Independent Set}. Let $G$ be
  our input graph, where $V(G) = \{c_1, \ldots, c_n\}$ and
  $E(G) = \{e_1, \ldots, e_L\}$, and let $h$ be the size of the
  desired independent set.  We construct the corresponding
  \textsc{AV-SwapApprovals-CI-Bribery} instance as follows.
  
  Let $B = 3h$ be our budget and let $t = B+1$ be a certain parameter
  (which we interpret as ``more than the budget''). We form the candidate
  set $C = V(G) \cup \{p\} \cup F \cup D$, where $p$ is the preferred
  candidate, $F$ is a set of $t(n+1)$ filler candidates, and $D$ is a set
  of $t$ dummy candidates. Altogether, there are $t(n+2)+n+1$
  candidates. We denote individual filler candidates by $\diamond$ and
  individual dummy candidates by $\bullet$; we fix the societal axis
  to be:
  \[
    \rhd =
    \underbrace{
      \overbrace{\diamond \cdots \diamond}^t c_1 \overbrace{\diamond \cdots \diamond}^t c_2 \diamond \cdots \diamond c_{n-1} \overbrace{\diamond \cdots \diamond}^t c_n \overbrace{\diamond \cdots \diamond}^t \overbrace{\bullet \cdots \bullet}^t p
    }_{t(n+2) + n + 1}
  \]
  For each positive integer $i$ and each candidate $c$, we write
  $\pprec_i(c)$ to mean the $i$-th candidate preceding $c$ in $\rhd$.
  Similarly, we write $\ssucc_i(c)$ to denote the $i$-th candidate
  after~$c$. We introduce the following voters:
  \begin{enumerate}
  \item For each edge $e_i = \{c_a,c_b\}$ we add an edge voter
    $v_{a,b}$ with approval set $[c_a,c_b]$. For each vertex
    $c_i \in V(G)$, we write $V(c_i)$ to denote the set of the three
    edge voters corresponding to the edges incident to $c_i$.

  \item Recall that $L = |E(G)|$. For each vertex candidate
    $c_i \in V(G)$, we add sufficiently many voters with approval set
    $[\pprec_{t}(c_i), \ssucc_{t}(c_i)]$, so that, together with the
    score from the edge voters, $c_i$~ends up with $L$ approvals.
   
  \item We add $L-3$ voters that approve $p$.
   
  \item For each group of $t$ consecutive filler candidates, we add
    $L+4t$ filler voters, each approving all the candidates in the
    group. %
  \end{enumerate}
  Altogether, $p$ has score $L-3$, all vertex candidates have score
  $L$, the filler candidates have at least $L+4t$ approvals each, and
  the dummy candidates have score $0$.  We set the committee size to
  be $k = t(n+1) + (n-h)+1$.  Prior to any bribery, each winning
  committee consists of $t(n+1)$ filler candidates and $(n-h)+1$
  vertex ones (chosen arbitrarily).  This completes our construction.
  
  Let us assume that $G$ has a size-$h$ independent set and denote it
  with $S$. For each $c_i \in S$ and each edge $e_t = \{c_i,c_j\}$, we
  bribe edge voter $v_{i,j}$ to move an approval from $c_i$ to a
  filler candidate right next to $c_j$. This is possible for each of
  the three edges incident to $c_i$ because $S$ is an independent set.
  As a consequence, each vertex from $S$ ends up with $L-3$ approvals.
  Thus only $n-h$ vertex candidates have score higher than $p$ and,
  so, there is a winning committee that includes $p$.
  
  For the other direction, let us assume that it is possible to ensure
  that $p$ belongs to some winning committee via a bribery of cost at
  most $B$. Let us consider the election after some such bribery was
  executed.  First, we note that all the filler candidates still have
  scores higher than $L+3t$ (this is so because decreasing a
  candidate's score always has at least unit cost and $B <
  t$). Similarly, $p$ still has score $L-3$ because increasing his or
  her score, even by one, costs at least $t$ (indeed, $p$ is separated
  from the other candidates by $t$~dummy candidates). Since $p$
  belongs to some winning committee, this means that at least $h$
  vertex voters must have ended up with score at most $L-3$. In fact,
  since our budget is $B = 3h$, a simple counting argument shows that
  exactly $h$ of them have score exactly $L-3$, and all the other ones
  still have score~$L$. Let~$S$ be the set of vertex candidates with
  score $L-3$. The only way to decrease the score of a vertex
  candidate $c_i$ from $L$ to $L-3$ by spending three units of the
  budget is to bribe each of the three edge voters from $V(c_i)$ to
  move an approval from $c_i$ to a filler candidate. However, if we
  bribe some edge voter $v_{i,j}$ to move an approval from $c_i$ to a
  filler candidate, then we cannot bribe that same voter to also move
  an approval away from $c_j$ (this would either cost more than $t$
  units of budget or would break the CI condition). Thus it must be
  the case that the candidates in $S$ correspond to a size-$h$
  independent set for~$G$.  \end{proof}

\hyphenation{fa-li-szew-ski}
For the VI domain, the complexity of our problem for unit prices
remains open, but for arbitrary prices we show that it
is $\np$-complete. Our proof works even for
the single-winner setting.
In the unrestricted
domain, the single-winner variant is in $\p$~\citep{fal:c:nonuniform-bribery}.

\begin{theorem}%
    \label{thm:vi-swaps}
      \textsc{AV-\$SwapApprovals-VI-Bribery} is $\np$-complete, even for
      the single-winner case (i.e., for committees of size one).
    \end{theorem}

                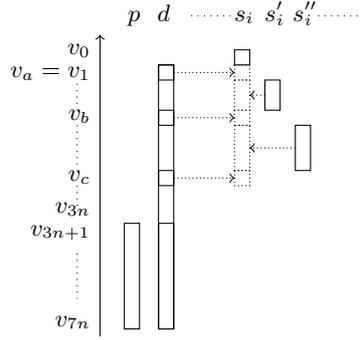
\begin{figure}
      \centering
      \begin{tikzpicture}
        \draw[->] (0,1) -- (0,5);
        \node[anchor=north east] at (0,5) {\small $v_0$};
        \node[anchor=north east] at (0,4.7) {\small $v_a = v_1$};
        \node[anchor=north east] at (0,4.125) {\small $v_b$};
        \node[anchor=north east] at (0,3.325) {\small $v_c$};
        \node[anchor=north east] at (0,2.9) {\small $v_{3n}$};        
        \node[anchor=north east] at (0,2.6) {\small $v_{3n+1}$};
        \node[anchor=north east] at (0,1.4) {\small $v_{7n}$};
        \draw[thin,dotted] (-0.3,4.35) -- (-0.3,4.0);
        \draw[thin,dotted] (-0.3,3.8) -- (-0.3,3.2);
        \draw[thin,dotted] (-0.3,3.0) -- (-0.3,2.7);
        \draw[thin,dotted] (-0.3,2.1) -- (-0.3,1.5);
    
        \node[anchor=south] at (0.45,5) {\small $p$};
        \draw (0.32,2.5) rectangle (0.52,1.1);
        \draw (0.77,2.5) rectangle (0.97,1.1);
        \node[anchor=south] at (0.92,5) {\small $d\phantom{p}$};
        \draw (0.77,4.6) rectangle (0.97,1.1);

        \node[anchor=south] at (1.98,5) {\small $s_i\phantom{p}$};
        \node[anchor=south] at (2.38,5) {\small $s'_i\phantom{p}$};
        \node[anchor=south] at (2.78,5) {\small $s''_i\phantom{p}$};

        \draw[dotted] (1.2,5.25) -- (1.77,5.25);
        \draw[dotted] (2.87,5.25) -- (3.5,5.25);

        \draw (1.77,4.8) rectangle (1.97,4.6);
        \draw (2.17,4.4) rectangle (2.37,4.0);
        \draw (2.57,3.8) rectangle (2.77,3.2);

        \draw (0.77,4.6) rectangle (0.97,4.4);
        \draw (0.77,4.0) rectangle (0.97,3.8);
        \draw (0.77,3.2) rectangle (0.97,3.0);    

        \draw[densely dotted] (1.77,4.6) rectangle (1.97,4.4);
        \draw[densely dotted] (1.77,4.0) rectangle (1.97,3.8);
        \draw[densely dotted] (1.77,3.2) rectangle (1.97,3.0);

                \draw[densely dotted] (1.77,4.4) rectangle (1.97,4.0);
        \draw[densely dotted] (1.77,3.8) rectangle (1.97,3.2);

        \draw[->, densely dotted] (0.97,4.5) -- (1.77,4.5);
        \draw[->, densely dotted] (0.97,3.9) -- (1.77,3.9);
        \draw[->, densely dotted] (0.97,3.1) -- (1.77,3.1);

        \draw[->, densely dotted] (2.17,4.2) -- (1.97,4.2);
        \draw[->, densely dotted] (2.57,3.5) -- (1.97,3.5);

      \end{tikzpicture}
      \caption{\label{fig:vi-swap}Illustration of the election from the proof of
        Theorem~\ref{thm:vi-swaps}, for the case where
        $\boldsymbol{S_i = \{x_a,x_b,x_c\}}$, where $\boldsymbol{a = 1 < b <c}$. Each row corresponds to a voter and each
        column corresponds to a candidate. Solid boxes show approvals prior
        bribery and dotted ones show approval moves.}
    \end{figure}

    \begin{proof}
      We give a reduction from \textsc{RX3C}. Let $I = (X,\calS)$ be
      an instance of \textsc{RX3C}, where
      $X = \{x_1, \ldots, x_{3n}\}$ is a universe and
      $\calS = \{S_1, \ldots, S_{3n}\}$ is a family of size-$3$
      subsets of $X$ (recall that each element from $X$ belongs to
      exactly three sets from $\calS$). We form a single-winner
      approval election with $7n+1$ voters
      $V = \{v_0, v_1, \ldots, v_{7n}\}$ and the following candidates:
      \begin{enumerate}
      \item We have the preferred candidate $p$ and the (to be
        defeated) current winner $d$.
      \item For each set $S_i \in \calS$ we have candidates $s_i$,
        $s'_i$, and $s''_i$.
      \end{enumerate}
      The approvals for these candidates, and the costs of moving
      them, are as follows (if we do not explicitly list the cost of
      moving some approval from a given candidate to another, then
      it is $+\infty$, i.e., this swap is impossible; the
      construction is illustrated in Figure~\ref{fig:vi-swap}):
      \begin{enumerate}       
      \item Candidate $p$ is approved by $4n$ voters, $v_{3n+1}, \ldots, v_{7n}$.
      \item Candidate $d$ is approved by $7n$ voters, $v_{1}, \ldots, v_{7n}$.
        For each set $S_i = \{x_a, x_b, x_c\}$, where $a < b < c$,
        the cost of moving $v_a$'s approval from $d$ to $s_i$ is $1$, and the costs of
        moving $v_b$'s and $v_c$'s approvals from $d$ to $s_i$ is $0$.
      \item For each set $S_i = \{x_a,x_b,x_c\}$, where $a < b < c$, we
        have the following approvals. Candidate $s_i$ is approved by
        voter $v_{a-1}$, candidate $s'_i$ is approved by voters
        $v_{a+1}, \ldots, v_{b-1}$, and candidate $s''_i$ is approved
        by voters $v_{b+1}, \ldots, v_{c-1}$. The cost of moving the
        approvals from $s'_i$ or from $s''_i$ to $s_i$ is $0$.
      \end{enumerate}
      One can verify that this election has the VI property for the
      natural order of the voters (i.e., for
      $v_0 \rhd \cdots \rhd v_{7n}$).  Candidate $d$ has $7n$
      approvals, $p$ has $4n$ approvals, and every other candidate has
      at most $3n+1$ approvals.  We claim that it is possible to
      ensure that $p$ becomes a winner of this election by
      approval-moves of cost at most $B = n$ (such that the election
      still has the VI property after these moves) if and only if $I$
      is a yes-instance of \textsc{RX3C}.

      For the first direction, let us assume that $I$ is a
      yes-instance and that $R \subseteq [3n]$ is a size-$n$ set such
      that $\bigcup_{i \in R}S_i = X$ (naturally, for each
      $t, \ell \in R$, sets $S_t$ and $S_\ell$ are disjoint). It is
      possible to ensure that~$p$ becomes a winner by making, for each
      $S_i = \{x_a, x_b, x_c\}$ such that $i \in R$ and $a < b < c$,
      the following swaps:
      \begin{enumerate}
      \item For each $j \in \{a,b,c\}$, we move $v_j$'s approval from
        $d$ to $s_i$ (the cost of moving $v_a$'s approval is $1$, the
        two other moves have cost $0$).
      \item For each $j \in \{a+1, \ldots, b-1\}$, we move $v_j$'s
        approval from $s'_i$ to $s_i$ (at cost 0).
      \item For each $j \in \{b+1, \ldots, c-1\}$, we move $v_j$'s
        approval from $s''_i$ to $s_i$ (at cost 0).
      \end{enumerate}
      In total, these moves cost $n$ and, since $R$ corresponds to a
      cover of~$X$, we have that: (a) $p$ is approved by $4n$ voters,
      (b) $d$ is approved by $4n$ voters, and (c) every other
      candidate is approved by at most $3n+1$ voters. Consequently,
      $p$ is among tied winners of this election.

      For the other direction, let us assume that there is a sequence
      of approval moves that costs at most $n$ and ensures that $p$ is
      a winner. Since all the moves of approvals from and to $p$ have
      cost $+\infty$, this means that every candidate ends up with at
      most $4n$ points. Thus $d$ loses at least $3n$ approvals.  No
      matter what swaps we do, for each $i \in [3n]$ each of $s_i$,
      $s'_i$ and $s''_i$ ends up with at most $3n+1$ approvals so we
      do not need to count their scores carefully (but we do need to
      take the VI condition into account for these candidates).

      Candidate $d$ can lose approvals only due to voters
      $v_1, \ldots, v_{3n}$ moving them to candidates in
      $\{s_1, \ldots, s_{3n}\}$. Let us consider some candidate $s_i$
      such that some voter $v_j$ moves an approval from~$d$ to~$s_i$
      and let $a < b < c$ be such that $S_i = \{x_a, x_b, x_c\}$. Due
      to the costs of moving approvals, it must be that
      $j \in \{a,b,c\}$. In fact, we claim that all three voters
      $v_a$, $v_b$, $v_c$ move approvals to~$s_i$, voters
      $v_{a+1}, \ldots, v_{b-1}$ move approvals from~$s'_i$ to~$s_i$,
      and voters $v_{b+1}, \ldots, v_{c-1}$ move approvals
      from~$s''_i$ to~$s_i$. This is so, because if those voters
      $v_a, \ldots, v_{j}$ would not move their approvals, then---due
      to the fact that $s_i$ is approved by voter $v_{a-1}$ (and this
      approval cannot move given our budget)---the approvals for $s_i$
      would not satisfy the VI property.  Further, voters
      $v_{j+1}, \ldots, v_c$ also need to move their approvals due to
      a counting argument: The cost of moving $v_a$'s approval from~$d$
      to $s_i$ is $1$. If we did not move $v_c$'s approvals from~$d$
      to $s_i$, then it would mean that (globally in our bribery)
      the average cost of moving an approval from $d$ to some
      candidate in $\{s_1, \ldots, s_{3n}\}$ would be higher than
      $\nicefrac{1}{3}$. But since our budget is $n$ and we need to
      move $3n$ approvals from $d$ to these candidates, this is
      impossible.

      Let $R = \{ i \in [3n] \mid$ some voter moves an approval from
      candidate $d$ to $s_i\}$. By the preceding paragraph, $R$
      contains $n$ elements and for each two $i, j \in R$ it must be
      that sets $S_i$ and $S_j$ are disjoint. Hence, $I$ is a
      yes-instance.
    \end{proof}

\section{Destructive Bribery}
\label{sec:destructive}

We conclude by considering destructive variants of our
problems, where the goal is to ensure that a given candidate,
often denoted $d$, does not belong to any winning committee (our proofs easily
trasfer to an alternative model, where one requires only that $d$ does not
belong to at least one winning committe; see~Appendix~\ref{app:other_models}) . We
use the same bribery actions, except that now we also consider a
variant of swapping approvals where we can only move approvals away
from $d$.

The destructive variant has been studied for the unrestricted setting
by \citet{yan:c:destructive-approval-bribery}, for the unpriced cases
of adding and deleting approvals. Thus we first establish its
complexity also for the priced cases and for swapping
approvals. %
The complexity stays the same as for the constructive variants (the
theorem below includes the results of
\citet{yan:c:destructive-approval-bribery} as special cases).

\begin{theorem}%
    \label{thm:destructive-unrestricted}
  For unit prices, all destructive variants of our bribery problems
  for the unrestricted setting are in $\p$. For arbitrary prices, the
  cases of adding and deleting approvals are in $\p$, but destructive
  variants \textsc{AV-\$SwapApprovalsAwayFromD-Bribery} and
  \textsc{AV-\$SwapApprovals-Bribery} are $\np$-complete.
\end{theorem}
    \begin{proof}
        {\bf Case of destructive \textsc{AV-SwapApprovals-Bribery}:}
        Let our input consist of an election $E = (C,V)$, committee size $k$, candidate $d \in C$, and budget $B$.
        We can assume that the solution does not contain any transitive swaps, i.e.,
        for $ \{c_1, c_2, c_3 \} \subseteq C$ a swap of an approval from $c_1$ to $c_2$ and from $c_2$ to $c_3$
        (because we could move from $c_1$ to $c_3$ directly).
        Let us consider some optimal solution.
        Assume that all swaps from $d$ were already executed and candidate $d$ is left with $t$ approvals.
        Now we treat $t$ as a fixed score of $d$ and based on the scores prior to making the remaining moves, we group all the candidates into three sets:
        \begin{enumerate}
            \item $L$ -- candidates with score lower than $t$. %
            \item $T$ -- candidates with score equal to $t$. %
            \item $G$ -- candidates with score greater than $t$. %
        \end{enumerate}
        Now let us consider the remaining approval moves in the solution.
        If there is any swap that moves an approval from candidate $x \in L$, or to some candidate $x \in L$, then we can replace it with a swap of an approval from $d$ to $x$.
        Such a move is possible because candidate $x$ has score lower than $d$, therefore a vote with an approval for $d$ and without an approval for $x$ must exist.
        It if were a move where $x$ was getting an approval, then this solution clearly remains correct.
        If $x$ were giving an approval to some other candidate $y$, then by lowering score of $d$ we achieve that distance of scores of $d$ and $y$ remain the same.
        
        If there is any swap that moves an approval from candidate $x \in T$ to candidate $y \in T$, then there is a vote with approval for $x$ and without approval for $y$.
        If this vote has an approval for $d$, then we can make the swap from $d$ instead of from $x$.
        If this vote does not have an approval for $d$, then there must be a vote where $d$ is approved but $x$ is not and we make the swap from $d$ to $x$ within that vote,
        achieving the same effect.
        
        If we implement all the above reasoning, we can be sure that there are left only swaps between candidates from the sets $T$ and~$G$.
        By performing these swaps, some candidates from $G$ may move to sets $L$ or $T$, and some candidates from $T$ may move to sets $L$ or $G$.
        However, if there is a single vote where approval for $d$ may be moved to any other candidate, then all candidates from $T$ are moved to $G$ and all candidates from $G$ remain in $G$.
        This is clearly a move which yields a solution at least as good as the initial one.
        If it is impossible to make a swap from $d$, it means that all voters who approve $d$ also approve all other candidates and all other candidates belong to either $T$ or $G$.
        To solve this case, repeatedly for each candidate from $G$ who has at least 2 approvals of advantage over $d$ we move an approval  to some candidate $x \in T$, moving $x$ to $G$ (this is correct because all candidates in $T$ are approved by exactly the same voters as $d$).
        
        The above reasoning shows that there is an optimal solution that consists only of two certain types of swaps,
        which leads to the following algorithm (by guessing $t$ we mean trying it for all possible values):
        \begin{enumerate}
            \item We guess the number of approvals $t$  that $d$ ends up with
            \item In a loop, until $d$ has exactly $t$ approvals,
                among all candidates who have less than $t+1$ approvals find the candidate $x$ that has the highest score
                and transfer an approval from $d$ to $x$.
                If there is such a candidate $x$, there exists a vote where we can move an approval from $d$ to $x$, because $d$ has at least $t+1$ approvals and $x$ has fewer.
            \item If at this point $d$ does not belong to any winning committee then we accept.
                Otherwise, either we guessed $t$ incorrectly or it is a situation where all the voters that approve $d$ also approve every other candidate.
            \item If we have not accepted yet, in a loop until the number of candidates with more than $t$ approvals is equal to $k$,
                find a candidate who has at least $t+2$ approvals and move one of these approvals to some other candidate who has exactly $t$ approvals (if such a candidate does not exist then reject for this value of $t$).
            \item If the number of swaps created during the algorithm exceeds $B$, then we reject (for this value of $t$) and accept in the opposite case.
        \end{enumerate}
        
        {\bf Case of destructive \textsc{AV-SwapApprovalsAwayFromD-Bribery}:}
        Let our input consist of an election $E = (C,V)$, committee size $k$, candidate $d \in C$, and budget $B$.
        As swaps from $d$ are the only available operation, our goal is to maximize the number of candidates with score greater than that of $d$ (without exceeding the budget).
        Let us assume that the solution exists and leaves $d$ with exactly $t$ approvals.
        There is no point in swapping approvals to candidates who, from the beginning, have more than $t$ approvals,
        and there is no point in giving any candidate approvals so that he or she has more than $t+1$ of them.
        
        Hence, given $t$, a simple greedy algorithm can decide whether the solution exists.
        It processes the candidates who have scores lower than $t+1$ in descending order of their scores
        and moves approvals from $d$ to them until either we exceed the budget or $d$ ends up with exactly $t$
        approvals. At this point, we check if $d$ is in some winning committee. If not, we accept. Otherwise, we reject.
        \smallskip

        {\bf Case of destructive \textsc{AV-\$AddApprovals-Bribery} and \textsc{AV-\$DeleteApprvals-Bribery}:}
        For the case of adding approvals, for each candidate we can independently compute
        a function saying how much it costs to increase his or her score to a given value (if a candidate already
        has score higher than $d$, then this value is equal to $0$).
        Then we select $k$ candidates for whom getting the score just above that of $d$
        is cheapest. If the sum of their costs is at most equal to the budget, we accept. Otherwise, we reject.
        
        For \textsc{AV-\$DelApprovals-Bribery}, it never makes sense to delete approvals from candidates
        other than $d$. Thus, we keep on deleting $d$'s cheapest approvals until he or she
        is not among the winners (we accept if doing so is within the budget and we reject otherwise).
        \smallskip

        {\bf Case of destructive \textsc{AV-\$SwapApprovalsAwayFromD-Bribery}
        and destructive \textsc{AV-\$SwapApprovals-Bribery}:}
        
        The reduction for~destructive \textsc{AV-\$SwapApprovalsAwayFromD-CI-Bribery}
        and destructive \textsc{AV-SwapApprovals-CI-Bribery}
				from~\Cref{thm:destructive-ci} (described later) can also be applied to
				this case.

        This follows from the fact that the candidate interval property does not constrain
        swapping approvals in elections, in which each voter supports only one
        candidate. The reduction from~\Cref{thm:destructive-ci} uses only
        such voters. Hence, the result therein holds also for destructive
        \textsc{AV-\$SwapAwayFromDApprovals-Bribery} and destructive
        \textsc{AV-\$SwapApprovals-Bribery}.
    \end{proof}

    For the VI case, we also obtain (or, fail to obtain)
    almost %
    the same results as in the constructive case (for
    \textsc{AV-\$SwapApprovals-VI-Bribery} we use the same proof as in
    the constructive case, except $d$ is the distinguished candidate
    and $p$ has one extra approval).

\begin{theorem}
  \label{thm:destructive-vi}
  Destructive variants of
  \textsc{AV-\$AddApprovals-VI-Bribery} and \textsc{AV-\$DelApprovals-VI-Bribery}
  are in $\p$.
  Destructive variant of \textsc{AV-\$SwapApprovals-VI-Bribery} is $\np$-complete.
\end{theorem}
\begin{proof}
        {\bf Case of destructive \textsc{AV-\$AddApprovals-VI-Bribery} and \textsc{AV-\$DeleteApprvals-VI-Bribery}:}
        For the case of adding approvals, for each candidate we can independently compute
        a function saying how much it costs to increase his or her score to a given value (if a candidate already
        has score higher than $d$, then this value is equal to $0$):
        This is just slightly more involved than the unrestricted case.
        For each candidate~$c$ we check for each pair of integers~$\{x,y\}$ with~$\scoreof{c}+x+y>\scoreof{d}$ how much it costs to add~$x$ approvals on top of the
        top-most original approval and to add~$y$ approvals below the bottom-most original approval.
        (This way, we check every possibility that keeps the interval of approvals intact.)
        Then we select $k$ candidates for whom getting the score above that of $d$
        is cheapest. If the sum of their costs is at most equal to the budget, we accept. Otherwise, we reject.
        
        For \textsc{AV-\$DelApprovals-Bribery}, it never makes sense to delete approvals from candidates
        other than $d$. Thus, we check for each pair of integers~$\{x,y\}$ with~$x+y<\scoreof{d}$ if deleting the $x$ top-most
        and the $y$ bottom-most approvals of~$d$ is possible with the given budget.
        We reject if this is not the case for any integer pair.
        \smallskip

        {\bf Case of destructive \textsc{AV-\$SwapApprovals-VI-Bribery}:}
        See proof of Theorem~\ref{thm:vi-swaps}.
\end{proof}

The case of CI preferences appears to be the most challenging one. Not only do we obtain fewer results than in the constructive setting, but those that we do obtain are less satisfying. Let us illustrate this with \textsc{AV-\$AddApprovals-CI-Bribery}. We show that the problem is $\np$-complete, but to do so, we use a somewhat unappealing trick. Namely, we include some voters who initially do not approve any candidates and we set their price functions so that we can choose one out of four candidates, possibly located far apart in the societal axis, to whom these voters add an approval. Since we did not put extra conditions on the price functions, this is formally correct, but is intuitively unappealing. We also need similar tricks in two other $\np$-completenss proofs in this section. For example, for the case of swapping approvals away from $d$ we use voters that approve only a single candidate, so we can move this approval arbitrarily (up to constraints implemented with the price function). Interestingly, if we required each voter to approve at least two candidates, the problem would be in $\p$.

\begin{theorem}%
  \label{thm:destructive-ci}
  Destructive variants of 
  \textsc{AV-\$AddApprovals-CI-Bribery}, \textsc{AV-\$SwapApprovalsAwayFromD-CI-Bribery},
  and 
  \textsc{AV-\$SwapApprovals-CI-Bribery} are $\np$-complete.
  Destructive variant of
  \textsc{AV-SwapApprovalsAwayFromD-CI-Bribery} is in $\p$.
\end{theorem}
\begin{proof}
  {\bf Case of destructive \textsc{AV-\$AddApprovals-CI-Bribery}:} We
  give a reduction from \textsc{RX3C}. Our input consists of a set
  $X = \{x_1, \ldots, x_{3k}\}$ and a family of sets
  $\calS = \{S_1, \ldots, S_{3n}\}$. We form an election with
  candidate set
  $X \cup \calS \cup \{d\} \cup \{a_1, \ldots\, a_{3n}\} \cup \{b_1,
  \cup, b_{3n}\}$, where we want $d$ to not be a part of any winning
  committee. The committee size is $k = 5n$.  The societal axis is
  $x_1 \rhd a_1 \rhd \cdots \rhd x_{3n} \rhd a_{3n} \rhd S_1 \rhd b_1 \rhd \cdots \rhd S_{3n} \rhd b_{3n} \rhd d$.
    We have voters so that, initially:
  \begin{enumerate}
  \item $d$ has 3 approvals,
  \item each $S_i$ has 1 approval,
  \item each $x_j$ has 3 approvals.
  \item all the $a_t$ and $b_t$ candidates have zero approvals.
  \end{enumerate}
  The voters that implement these scores have such high prices, so we cannot add approvals to them. 
  Further, for each set $S_t = \{x_i, x_j, x_k\}$ we form three \emph{solution} voters. Each of these three voters has an empty approval set and adding each approval costs $+\infty$, except for adding approvals for $S_t$, $x_i$, $x_j$, and $x_k$, which have unit cost.
  The budget is $9n$.  
  Prior to any bribery, $d$ is in some winning committees and we need to get $5n$ candidates to have score above $3$ to prevent this.
    
  If there is an exact cover of $X$ with $n$ sets from $\calS$, then we can add approvals so that for each
  set from the cover, the three corresponding solution voters approve its members (this costs $3n$ and ensures that each candidate from $X$ has $4$ approvals). For the $2n$ remaining sets (not forming the cover), each solution voter adds an approval
  for its corresponding set (this costs $6n$ and ensures that $2n$ set candidates have $4$ approvals each). Consequently, there are $5n$ candidates with score higher than $d$ and, so, $d$ does not belong to any winning committee.
   
  Next, assume that there is a bribery after which $d$ does not belong to any winning committee. After performing
  this bribery, there are $x$ set candidates that have at least $4$ approvals each, and $y$ candidates from $X$ that have at least $4$ approvals each. Adding these approvals has cost:
  \[
    3x + y \leq 9n
  \]
  This bribery ensures that there are at least $5n$ candidates with score higher than that of $d$, so we know that
  $x+y \geq 5n$. Together, this implies that $x \leq 2n$, so $y \geq 3n$, which means that $y = 3n$ and, consequently, $x = 2n$.
  So there are $2n$ sets that get (at least) $4$ approvals each. But for these sets, we cannot give approvals to their members of $X$. So the remaining $n$ sets must form a cover.
  \smallskip

        {\bf Case of destructive \textsc{AV-\$SwapApprovalsAwayFromD-CI-Bribery}
        and destructive \textsc{AV-\$SwapApprovals-CI-Bribery}:}

        We reduce from \textsc{Independent Set} on regular graphs, where given a
        (simple) $\Delta$-regular graph~$G=(U,E)$ and a positive integer~$h$, we ask
        whether there is a set of at least~$h$~vertices, in which no two vertices share
        an edge; the sought set is called an~\emph{independent set}. For some~$u \in U$,
        let~$E(u) = \{e: u \in e\}$ be the set of edges \emph{incident} to~$u$, that is,
        the edges that have an endpoint in~$u$. In particular, for each vertex~$u$ of
        some $\Delta$-regular graph, it holds that~$|E(u)| = \Delta$. In our proof, it is
        convenient to define an independent set on $\Delta$-regular graphs as a set~$S$
        of vertices of the graph that have pairwise disjoint sets of incident edges,
        that is, $\left| \bigcup_{u \in S} E(u) \right| = \Delta |S|$.
        
        Consider an instance~$I = (G, h)$ of \textsc{Independent Set}. Let~$G = (U, E)$ be a
        $\Delta$-regular graph with the set~$U = \{u_1, u_2, \ldots u_n\}$~of vertices
        and the set~$E = \{e_1, e_2, \ldots, e_m\}$~of edges.
        
        We construct an approval election with committee size~$k = h$ and the candidate
        set~$C$. Starting with an empty set~$C$, for each vertex~$u \in U$, we add a
        candidate~$c(u)$ and for each edge~$e \in E$ we add a candidate~$c(e)$.
        Finally, we add a~\emph{despised candidate}~$d$ obtaining, in total, $|C| = n +
        m + 1$~candidates. The collection of voters~$V$ is build up by two
        groups of voters. The first group consists of $(\Delta  -1)$~voters approving~$d$;
        these voters cannot be bribed. In the second group, for each edge~$e = \{u, u'\} \in E$,
        there is a voter~$v(e)$ that approves~$d$ whose approval can be swapped to
        support one of~$c(u)$, $c(u')$, and $c(e)$ at cost $0$; any other swap is
        impossible. Last but not least, we set the budget~$B$ of our new instance to
        be~$0$. Accordingly, we implement the forbidden swaps that we described above by
        giving them cost~$1$. This concludes the construction of the instance~$I'$
        of destructive \textsc{AV-\$SwapAwayFromDApprovals-CI-Bribery} that we reduce
        instance~$I$ to.
        
        To prove the reduction's correctness, let us first assume that~$I$ admits an
        independent set~$S$ of size at least~$h$. Without loss of generality, we assume
        that~$S = \{u_1, u_2, \ldots, u_h\}$ (indeed, each subset of an independent set
        is an independent set itself and the vertex names can be relabeled). We construct
        a solution to~$I'$ by swapping the votes as follows. Consider a vertex $u \in
        S$. For each incident edge~$e \in E(u)$ of $u$, we make~$v(e)$ support~$c(u)$
        instead of~$d$; by the construction, it comes at cost zero. For every remaining
        edge, say~$e \in E \setminus \bigcup_{u \in S} E(u) $, we swap the approval of
        the corresponding vote~$v(e)$ away from~$d$ and give it to~$c(e)$; this action
        also costs zero budget units. As a result, $d$ gets~$\Delta -1$~approvals,
        candidates~$c(u_1)$ to~$c(u_h)$ get~$\Delta$~approvals each, and the remaining
        candidates get either one or zero approvals. Hence, the winning committee of
        size~$h$ does not contain~$d$ as it exactly consists of candidates~$c(u_1)$
        to~$c(u_h)$ and the budget spent is $0$.
        
        To prove the opposite direction, assume that there is a bribery action leading
        to an election~$E'$ in which all winning committees of size~$h$ do not
        contain~$d$. We claim that each of these committees consists of candidates
        corresponding to the vertices that form an independent set of
        size~$h$ in~$G$. However, for ease of presentation and without loss of
        generality, we fix one such committee~$S$. Observe that the minimal score of~$d$
        is~$\Delta-1$ as there are exactly $\Delta-1$~voters that cannot be bribed that
        approve~$d$. Hence, there are at least~$h$~candidates with having at
        least~$\Delta$~approvals in~$E'$.  Note that, by construction, no
        candidate~$c(e)$ corresponding to an edge~$e \in E$ has a score higher than one
        in~$E$. Let us fix some~$u \in U$. We now show that
        candidate~$c(u)$~corresponding to~$u \in U$ cannot obtain more than
        $\Delta$~approvals in~$E$. Indeed, $c(u)$ can only get approvals by bribing
        voters corresponding to the edges incident to~$u$. Since $u$ belongs to a
        $\Delta$-regular input graph, it has exactly~$\Delta$ incident edges. So, there
        is at least~$h$~candidates (corresponding to vertices of~$G$) with
        exactly~$\Delta$~approvals in~$E$. From the discussion above it is also clear
        that to make~$u$ get $\Delta$~approvals, one needs to bribe all voters
        corresponding to the incident edges of~$u$ and, obviously, each of these voters
        contribute to the score of exactly one candidate. Hence, to make all candidates
        of the winning committee~$S$ have score~$\Delta$, it must hold that for each
        candidate~$c(u) \in S$, representing vertex~$u$, there is a disjoint set of~$\Delta$
        voters, representing edges incident to~$u$, that has to be bribed. So the set of
        vertices corresponding to members of~$S$ must be an independent set.
        
        This reduction clearly works in polynomial time and forms a parameterized
        reduction with respect to the value of the committee size, which concludes the
        proof. Since instance~$I'$ contains votes that initially approve only one
        candidate, the argument is also correct
        for the destructive~\textsc{AV-\$SwapApprovals-CI-Bribery} problem.

     {\bf Case of destructive \textsc{AV-SwapApprovalsAwayFromD-CI-Bribery}:}
  We solve the problem via dynamic programming.
  The crucial observations are as follows.
  First, each voter who approves~$d$ and some other candidate(s) either to the left, or to the right
  can move the approval to just one candidate at unit price.
  For each candidate~$c_i$, let~$Y_i$ denote the number of voters that approve~$d$
  and every candidate between~$d$ and~$c_i$ in the societal order.
  Second, for voters who only approve~$d$ we can move the approval to any other candidate at unit cost.
  Let~$X$ denote the total number of such voters.

  First, we guess the final score~$s^*$ of~$d$ after bribery: it is at least $\scoreof{d}-B$ and at most $\scoreof{d}-1$.
  (Note that we need to swap away~$\scoreof{d}-s^*<B$ approvals from~$d$.)
  We have the following binary table~$T[i,k',B',x] \in \{0,1\}$,
  where $T[i,k',B',x]=1$ when it is possible move approvals from~$d$ to the first~$i$ candidates in the societal order,
  while ensuring that~$k'$ of them have score larger than~$s^*$, exactly~$B'$ approvals where swapped away from~$d$,
  and in total~$x$ of these approvals have been swapped by voters only approving~$d$.
  We initialize the table by setting first setting $T[1,*,*,*]:=0$ and then updating:
  (a) $T[1,0,B'',x']:=1$ if~$\scoreof{c_1}+B''\le s^*$, $x'\le X$, $x'\le B''$, and~$B''-x'\le Y_1$,
  (b) $T[1,1,B'',x']:=1$ if~$\scoreof{c_1}+B''> s^*$, $x'\le X$, $x'\le B''$, and~$B''-x'\le Y_1$.
  Cases under~(a) capture all possibilities to move approvals from~$d$ to the leftmost candidate
  while keeping its score to at most~$s^*$ while cases under~(b) capture all possibilities to move approvals
  from~$d$ to the leftmost candidate while ensuring its score to be larger than the score~$s^*$ of~$d$.
  Entries for~$i>1$ (with $i$~increasing from~$2$ to~$m$) are computed as follows.
  Set $T[i,*,*,*]:=0$ and then update:
  (a) $T[i,k',B'',x']:=1$ if there are some~$B^*$ and~$x^*$ with $0 \le B^* \le B''$ and $0 \le x^* \le x'$ such that
  $\scoreof{c_i}+B^*\le s^*$, $x'\le X$, $x'\le B''$, $B^*-x^*\le Y_i$, and $T[i-1,k',B''-B^*,x'-x^*]=1$.
  (b) $T[i,k',B'',x']:=1$ if there are some~$B^*$ and~$x^*$ with $0 \le B^* \le B''$ and $0 \le x^* \le x'$ such that
  $\scoreof{c_i}+B^* > s^*$, $x'\le X$, $x'\le B''$, $B^*-x^*\le Y_i$, and $T[i-1,k'-1,B''-B^*,x'-x^*]=1$.
  Similar to the initialization, cases under~(a) capture all possibilities to move approvals from~$d$ to the first~$i$ candidates
  while keeping the score of~$c_i$ to at most~$s^*$ while cases under~(b) capture all possibilities to move approvals
  from~$d$ to the first~$i$ candidate while ensuring that the score of~$c_i$ is than the score~$s^*$ of~$d$.
  
  Finally, we have if yes-instance if there is some entry~$T[m,k^*,\scoreof{d}-s^*,x^*]=1$ with~$x^*<X$ and~$k^* \ge k$.
\end{proof}

\section{Summary}
We have studied bribery in multiwinner approval elections, for the
case of candidate interval (CI) and voter interval (VI)
preferences. Depending on the setting, our problem can either be
easier, harder, or equally difficult as in the unrestricted domain.
It would be interesting to extend our work by considering different
voting rules (in particular, the Approval-Based Chamberlin--Courant
rule~\citep{cha-cou:j:cc,pro-ros-zoh:j:proportional-representation,bet-sli-uhl:j:mon-cc})
and by seeking parameterized complexity results.

\section{Acknowledgements}
We want to thank the reviewers for their helpful comments.
This project has received funding from the European 
Research Council (ERC) under the European Union's Horizon 2020 
research and innovation programme (grant agreement No 101002854).
DK acknowledges the support of the GA\v{C}R project No.~22-19557S.

\begin{center}
\includegraphics[width=2.5cm]{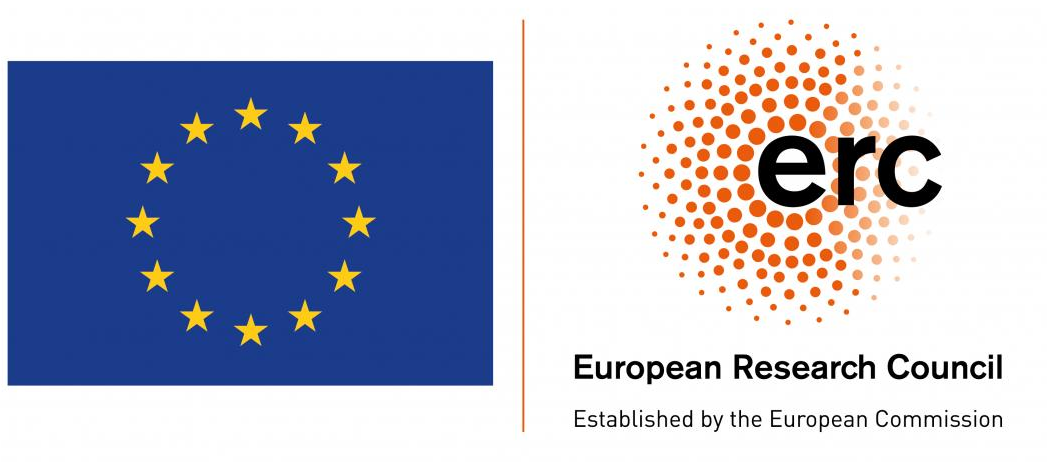}
\end{center}

\bibliographystyle{abbrvnat}
\bibliography{bib-arxiv}

    \clearpage
    \appendix
   
  \section*{Supplementary Material}

	\section{Adapting Proofs to Other Winning Models}\label{app:other_models}
    In the constructive case, we will adapt our results to the case in which the preferred candidate is in all committees. In the destructive case, after the adaptation, the despised candidate is considered to lose the election if there is at least one winning committee without this candidate
   \subsection{Constructive}
      \begin{description}
       \item[\Cref{thm:add-approvals-vi}:] Replace ``some   '' by ``all'' in the definition of $s$.
       \item[\Cref{thm:add-approvals-ci}:] Replace ``end up with highe'' by ``end at with at least as hig'' in the definition of $t_\ell$, and $t_r$, and adjust the definition Step 3 accordingly.
       \item[\Cref{thm:delete-approvals-vi}:] Definition of ``superior candidate'' must be adjusted to contain also candidates with score equals to the score of $p$. 
       \item[\Cref{thm:delete-approvals-ci}:] Add one more fixed approval to~$p$.
       \item[\Cref{thm:swap-approvals-to-p-ci}:] The definition of $S$ needs to be adjusted to contain also candidates with the same score as $p$. The rest remains the same.
       \item[\Cref{thm:swap-approvals-to-p-vi}:] The definition of dangerous candidates needs to be adjusted so that it also contains those with same score as $p$. Again, the rest can be kept. The rest remains the same.
       \item[\Cref{thm:swap-approvals-ci}:] Add another voter that only approves p. This way, the $\Rightarrow$ direction carries over ``as is.'' In the $\Leftarrow$ direction we have to observe that bribing this new voter brings nothing good. Indeed, it is the case since when we do so, then we do not have enough budget to do other changes that should decrease the scores of other candidates to be lower than the score of $p$.
       \item[\Cref{thm:vi-swaps}:] We add one more voter approving only~$p$ such that the voter cannot be bribed.

     \end{description}
   \subsection{Desctructive}
     \begin{description}
       \item[\Cref{thm:destructive-unrestricted}:] In the algorithmic results, modify defining groups of candidates to reflect that now it is enough that we have $k$ candidates (other than $d$) with score at least $t$ instead of at least $t+1$. The reduction is the same as that in~\Cref{thm:destructive-ci}.
       \item[\Cref{thm:destructive-vi}:] For the algorithmic result, we simply need to select $k$ candidates for wchich achieving the score of~$d$ is the cheapest. The reduction is the same as that in~
       \item[\Cref{thm:destructive-ci}:] In the case of adding approvals with arbitrary prices, candidate $d$ should have score~$4$ in point~(1) of the construction. The rest remains the same (up to some fluffy text). In the cases of swapping approvals away from~$d$ with arbitrary prices assuming CI and swapping approvals in the CI  domain with arbitrary prices: One should have~$\Delta$~voters supporting~$d$ (that cannot be bribed) instead of~$\Delta - 1$~of them. In the last case, swapping away from $d$ for unit prices in CI, the~$k'$ in the definition of the recurrent function needs to be changed to mean the number of candidates with score at least~$s^*$ (some details of the computation of the function have to be adjusted).
     \end{description}

\end{document}